\documentclass[prr,superscriptaddress,notitlepage,twocolumn,footinbib]{revtex4-2}
\usepackage{amsmath,amsfonts,amssymb}
\usepackage{mathtools}
\usepackage{multirow}
\usepackage[utf8]{inputenc}
\usepackage[T1]{fontenc}

\usepackage{bm}
\usepackage{times} 
\usepackage{mathtools}
\usepackage{amsmath}
\usepackage[shortlabels]{enumitem}

\usepackage{graphicx,epic,eepic,epsfig,amsmath,latexsym,amssymb,verbatim,color}
 
\usepackage{amsfonts}       
\usepackage{nicefrac}       

\usepackage{amsmath}
\usepackage{bbm}

\usepackage{float}
\usepackage{tikz}
\usetikzlibrary{chains}
\usetikzlibrary{fit}

\usepackage{epsfig}
\usetikzlibrary{shapes.symbols,patterns} 

\usepackage[strict]{changepage}
\usepackage{hyperref}
\hypersetup{colorlinks=true,citecolor=blue,linkcolor=blue,filecolor=blue,urlcolor=blue,breaklinks=true}

\usepackage[marginal]{footmisc}
\usepackage{url}
\usepackage{theorem}

\newtheorem{proposition}{Proposition}
\newtheorem{lemma}[proposition]{Lemma}

\newtheorem{definition}[proposition]{Definition}

\newtheorem{theorem}[proposition]{Theorem}

\def\squareforqed{\hbox{\rlap{$\sqcap$}$\sqcup$}}
\def\qed{\ifmmode\squareforqed\else{\unskip\nobreak\hfil
\penalty50\hskip1em\null\nobreak\hfil\squareforqed
\parfillskip=0pt\finalhyphendemerits=0\endgraf}\fi}
\def\endenv{\ifmmode\;\else{\unskip\nobreak\hfil
\penalty50\hskip1em\null\nobreak\hfil\;
\parfillskip=0pt\finalhyphendemerits=0\endgraf}\fi}
\newenvironment{proof}{\noindent \textbf{{Proof~} }}{\hfill $\blacksquare$}

\newcounter{remark}
\newenvironment{remark}[1][]{\refstepcounter{remark}\par\medskip\noindent%
\textbf{Remark~\theremark #1} }{\medskip}

\newcounter{example}

\mathchardef\ordinarycolon\mathcode`\:
\mathcode`\:=\string"8000
\def\vcentcolon{\mathrel{\mathop\ordinarycolon}}
\begingroup \catcode`\:=\active
  \lowercase{\endgroup
  \let :\vcentcolon
  }

\usepackage{cleveref}
\usepackage{graphicx}
\usepackage{xcolor}

\RequirePackage[framemethod=default]{mdframed}
\newmdenv[skipabove=7pt,
skipbelow=7pt,
backgroundcolor=darkblue!15,
innerleftmargin=5pt,
innerrightmargin=5pt,
innertopmargin=5pt,
leftmargin=0cm,
rightmargin=0cm,
innerbottommargin=5pt,
linewidth=1pt]{tBox}

\newmdenv[skipabove=7pt,
skipbelow=7pt,
backgroundcolor=blue2!25,
innerleftmargin=5pt,
innerrightmargin=5pt,
innertopmargin=5pt,
leftmargin=0cm,
rightmargin=0cm,
innerbottommargin=5pt,
linewidth=1pt]{dBox}

\newmdenv[skipabove=7pt,
skipbelow=7pt,
backgroundcolor=darkred!15,
innerleftmargin=5pt,
innerrightmargin=5pt,
innertopmargin=5pt,
leftmargin=0cm,
rightmargin=0cm,
innerbottommargin=5pt,
linewidth=1pt]{rBox}

\newmdenv[skipabove=7pt,
skipbelow=7pt,
backgroundcolor=darkkblue!15,
innerleftmargin=5pt,
innerrightmargin=5pt,
innertopmargin=5pt,
leftmargin=0cm,
rightmargin=0cm,
innerbottommargin=5pt,
linewidth=1pt]{sBox}
\definecolor{darkblue}{RGB}{0,76,156}
\definecolor{darkkblue}{RGB}{0,0,153}
\definecolor{blue2}{RGB}{102,178,255}
\definecolor{darkred}{RGB}{195,0,0}

\newcommand{\nc}{\newcommand}
\nc{\rnc}{\renewcommand}
\nc{\lbar}[1]{\overline{#1}}
\nc{\bra}[1]{\langle#1|}
\nc{\ket}[1]{|#1\rangle}
\nc{\ketbra}[2]{|#1\rangle\!\langle#2|}
\nc{\braket}[2]{\langle#1|#2\rangle}

\nc{\proj}[1]{| #1\rangle\!\langle #1 |}
\nc{\avg}[1]{\langle#1\rangle}
\nc{\rank}{\operatorname{Rank}}
\nc{\smfrac}[2]{\mbox{$\frac{#1}{#2}$}}
\nc{\tr}{\operatorname{Tr}}
\nc{\ox}{\otimes}
\nc{\dg}{\dagger}
\nc{\dn}{\downarrow}
\nc{\cA}{{\cal A}}
\nc{\cB}{{\cal B}}
\nc{\cC}{{\cal C}}
\nc{\cD}{{\cal D}}
\nc{\cE}{{\cal E}}
\nc{\cF}{{\cal F}}
\nc{\cG}{{\cal G}}
\nc{\cH}{{\cal H}}
\nc{\cI}{{\cal I}}
\nc{\cJ}{{\cal J}}
\nc{\cK}{{\cal K}}
\nc{\cL}{{\cal L}}
\nc{\cM}{{\cal M}}
\nc{\cN}{{\cal N}}
\nc{\cO}{{\cal O}}
\nc{\cP}{{\cal P}}
\nc{\cQ}{{\cal Q}}
\nc{\cR}{{\cal R}}
\nc{\cS}{{\cal S}}
\nc{\cT}{{\cal T}}
\nc{\cU}{{\cal U}}
\nc{\cV}{{\cal V}}
\nc{\cX}{{\cal X}}
\nc{\cY}{{\cal Y}}
\nc{\cZ}{{\cal Z}}
\nc{\cW}{{\cal W}}
\nc{\csupp}{{\operatorname{csupp}}}
\nc{\qsupp}{{\operatorname{qsupp}}}
\nc{\var}{{\operatorname{var}}}
\nc{\rar}{\rightarrow}
\nc{\lrar}{\longrightarrow}
\nc{\polylog}{{\operatorname{polylog}}}
\nc{\wt}{{\operatorname{wt}}}
\nc{\av}[1]{{\left\langle {#1} \right\rangle}}
\nc{\supp}{{\operatorname{supp}}}

\nc{\argmin}{{\operatorname{argmin}}}

\def\a{\alpha}

\def\i{\mathbf{i}}

\def\x{\xi}

\nc{\BS}{{{\mathbb S}}}
\nc{\RR}{{{\mathbb R}}}
\nc{\CC}{{{\mathbb C}}}
\nc{\FF}{{{\mathbb F}}}
\nc{\NN}{{{\mathbb N}}}
\nc{\ZZ}{{{\mathbb Z}}}
\nc{\PP}{{{\mathbb P}}}
\nc{\QQ}{{{\mathbb Q}}}
\nc{\UU}{{{\mathbb U}}}
\nc{\EE}{{{\mathbb E}}}
\nc{\id}{{\operatorname{id}}}

\nc{\CHSH}{{\operatorname{CHSH}}}

\nc{\<}{\langle}
\rnc{\>}{\rangle}
\nc{\rU}{\mbox{U}}

\nc{\ob}[1]{#1}

\nc{\SEP}{{\text{\rm SEP}}}
\nc{\NS}{{\text{\rm NS}}}
\nc{\LOCC}{{\text{\rm LOCC}}}
\nc{\PPT}{{\text{\rm PPT}}}
\nc{\EXT}{{\text{\rm EXT}}}
\nc{\Sym}{{\operatorname{Sym}}}

\nc{\ERLO}{{E_{\text{r,LO}}}}
\nc{\ERLOCC}{{E_{\text{r,LOCC}}}}
\nc{\ERPPT}{{E_{\text{r,PPT}}}}
\nc{\ERLOCCinfty}{{E^{\infty}_{\text{r,LOCC}}}}
\nc{\Aram}{{\operatorname{\sf A}}}

\newcommand{\Choi}{Choi-Jamio\l{}kowski }

\usepackage{tikz}
\usepackage{hyperref}
\hypersetup{colorlinks=true,citecolor=blue,linkcolor=blue,filecolor=blue,urlcolor=blue,breaklinks=true}

\makeatletter
\def\grd@save@target#1{%
  \def\grd@target{#1}}
\def\grd@save@start#1{%
  \def\grd@start{#1}}
\tikzset{
  grid with coordinates/.style={
    to path={%
      \pgfextra{%
        \edef\grd@@target{(\tikztotarget)}%
        \tikz@scan@one@point\grd@save@target\grd@@target\relax
        \edef\grd@@start{(\tikztostart)}%
        \tikz@scan@one@point\grd@save@start\grd@@start\relax
        \draw[minor help lines,magenta] (\tikztostart) grid (\tikztotarget);
        \draw[major help lines] (\tikztostart) grid (\tikztotarget);
        \grd@start
        \pgfmathsetmacro{\grd@xa}{\the\pgf@x/1cm}
        \pgfmathsetmacro{\grd@ya}{\the\pgf@y/1cm}
        \grd@target
        \pgfmathsetmacro{\grd@xb}{\the\pgf@x/1cm}
        \pgfmathsetmacro{\grd@yb}{\the\pgf@y/1cm}
        \pgfmathsetmacro{\grd@xc}{\grd@xa + \pgfkeysvalueof{/tikz/grid with coordinates/major step}}
        \pgfmathsetmacro{\grd@yc}{\grd@ya + \pgfkeysvalueof{/tikz/grid with coordinates/major step}}
        \foreach \x in {\grd@xa,\grd@xc,...,\grd@xb}
        \node[anchor=north] at (\x,\grd@ya) {\pgfmathprintnumber{\x}};
        \foreach \y in {\grd@ya,\grd@yc,...,\grd@yb}
        \node[anchor=east] at (\grd@xa,\y) {\pgfmathprintnumber{\y}};
      }
    }
  },
  minor help lines/.style={
    help lines,
    step=\pgfkeysvalueof{/tikz/grid with coordinates/minor step}
  },
  major help lines/.style={
    help lines,
    line width=\pgfkeysvalueof{/tikz/grid with coordinates/major line width},
    step=\pgfkeysvalueof{/tikz/grid with coordinates/major step}
  },
  grid with coordinates/.cd,
  minor step/.initial=.2,
  major step/.initial=1,
  major line width/.initial=2pt,
}
\makeatother

\usepackage{thmtools}
\usepackage{thm-restate}
\usepackage{etoolbox}
\makeatletter
\def\problem@s{}
\newcounter{problems@cnt}

\newcommand{\allproblems}{\problem@s}
\makeatother
\definecolor{beamer}{rgb}{0.2,0.2,0.7}
\definecolor{colorone}{rgb}{1,0.36,0.03}
\definecolor{colortwo}{rgb}{0.4,0.77,0.17}
\definecolor{colorthree}{rgb}{0.01,0.51,0.93}
\definecolor{colorfour}{rgb}{0.47,0.26,0.58}
\definecolor{colorfive}{rgb}{0.12,0.55,0.16}
\usepackage{tcolorbox}
\usepackage{relsize}
\usepackage{graphicx}
\usepackage{booktabs}
\usepackage{amsmath}
\usepackage{tikz}
\usetikzlibrary{quantikz}
\nc{\st}{\text{subject to} \ }
\nc{\supre}{\text{supremum} \ }
\nc{\sdp}{\text{sdp}}
\usepackage{array}

\newcommand{\sgn}{{\rm sgn}}
\allowdisplaybreaks
\newcommand{\tc}{{\rm TC}}
\newcommand{\idop}{\mathbbm{1}} 
\newcommand{\lrp}[1]{\left( #1 \right)}
\newcommand{\lrb}[1]{\left[ #1 \right]}
\newcommand{\lrc}[1]{\left\{ #1 \right\}}

\newcommand{\lrV}[1]{\left\| #1 \right\|}

\begin{document}
\title{Power of quantum measurement in simulating unphysical operations}

\author{Xuanqiang Zhao}
\affiliation{QICI Quantum Information and Computation Initiative, Department of Computer Science, The University of Hong Kong, Pokfulam Road, Hong Kong}

\author{Lei Zhang}
\affiliation{Thrust of Artificial Intelligence, Information Hub,\\ The Hong Kong University of Science and Technology (Guangzhou), Guangdong 511453, China}

\author{Benchi Zhao}
\affiliation{Thrust of Artificial Intelligence, Information Hub,\\ The Hong Kong University of Science and Technology (Guangzhou), Guangdong 511453, China}

\author{Xin Wang}
\email{felixxinwang@hkust-gz.edu.cn}
\affiliation{Thrust of Artificial Intelligence, Information Hub,\\ The Hong Kong University of Science and Technology (Guangzhou), Guangdong 511453, China}

\begin{abstract}
The manipulation of quantum states through linear maps beyond quantum operations has many important applications in various areas of quantum information processing.
Current methods simulate unphysical maps by sampling physical operations according to classically determined probability distributions.
In this work, we show that using quantum measurement instead leads to lower simulation costs for general Hermitian-preserving maps.
Remarkably, we establish the equality between the simulation cost and the well-known diamond norm, thus closing a previously known gap and assigning diamond norm a universal operational meaning for all Hermitian-preserving maps.
We demonstrate our method in two applications closely related to error mitigation and quantum machine learning, where it exhibits a favorable scaling.
These findings highlight the power of quantum measurement in simulating unphysical operations, in which quantum interference is believed to play a vital role. Our work paves the way for more efficient sampling techniques and has the potential to be extended to more quantum information processing scenarios.
\end{abstract}

\date{\today}
\maketitle

\section{Introduction}
Quantum measurement is the key operation that gives rise to the probabilistic nature of quantum mechanics. It can be viewed as a quantum way of sampling outcomes from a probability distribution governed by the Born rule.
The task of quantum random sampling~\cite{lund2017quantum, hangleiter2023computational} was conceived to demonstrate a computational advantage of quantum computers based on the observation that quantum measurement following certain quantum computations is difficult to simulate classically~\cite{terhal2002adaptive}.
Due to its relative simplicity, quantum random sampling has led to experimental demonstration of quantum advantage on noisy intermediate-scale quantum devices available nowadays~\cite{arute2019quantum, zhong2020quantum, wu2021strong, madsen2022quantum}. However, the task itself is proposed for the sole purpose of showing quantum advantage and therefore lacks practical meaning.
On real-world sampling problems, the benefits of quantum measurement are yet to be discovered.

A practical task of broad interest where sampling naturally arises is the simulation of unphysical linear maps, which is an essential subroutine in many quantum information processing tasks.
Though quantum operations are limited to completely positive and trace-preserving (CPTP), or more generally, completely positive and trace-non-increasing (CPTN) maps due to the physicality requirement that they map quantum states to quantum states, many operations beyond CPTN maps are of great importance from both theoretical and practical perspectives.
For example, positive but not completely positive maps, such as partial transposition~\cite{peres1996separability, HORODECKI19961}, are widely used to characterize and detect entanglement in quantum states~\cite{guhne2009entanglement}. Maps that are even non-positive are encountered in the mitigation of errors on quantum devices~\cite{Temme2017, Jiang2020, rossini2023single}. These crucial applications motivate the research of realizing unphysical operations, primarily Hermitian-preserving maps.

Among a few different paths to realizing Hermitian-preserving maps~\cite{horodecki2002method, Regula2021a, Jiang2020, wei2023realizing}, quasi-probability decomposition (QPD) has been a popular method due to its favorable memory requirement and easy implementability.
The idea of QPD is to decompose an unphysical operation $\cE$ into a linear combination of physical operations. Through sampling physical operations and classically post-processing measurement outcomes, this decomposition allows the computation of the expectation value of any observable with respect to any state transformed by a series of physical or unphysical operations with $\cE$ being an intermediate step.

While QPD has enjoyed great success in a variety of tasks~\cite{Temme2017, endo2018practical, takagi2021optimal, piveteau2022quasiprobability, van2023probabilistic, wang2022detecting, yuan2023virtual}, it is still conservative for that it samples physical operations according to a fixed probability distribution generated classically.
Quantum mechanics allows a more general way of generating a probability distribution and a tool to sample from it, that is, by quantum measurement.
Thus, it is then natural to ask: does quantum measurement bring any advantage over QPD in simulating maps beyond quantum operations?

In this work, we give an affirmative answer to this question by utilizing a quantum instrument to sample physical operations. It has been shown that every Hermitian-preserving map can be statistically decomposed to a quantum instrument with appropriate post-processing~\cite{buscemi2013direct}. Here, we rigorously study the simulation cost and the optimal protocol for simulating any Hermitian-preserving map using quantum instrument.
We show that one quantum instrument is all you need to achieve an optimal simulation cost, which is significantly lower than the one attained by QPD in some cases. Meanwhile, we prove that the simulation cost of any Hermitian-preserving map is equal to the map's diamond norm, generalizing the result in Ref.~\cite{Regula2021a} and endowing diamond norm with a universal operational meaning in the simulation of Hermitian-preserving maps. To demonstrate the advantage of quantum instruments, we consider applications in retrieving faithful information from noisy states and extracting entries from a state's density matrix, which are tasks closely related to error mitigation and quantum machine learning.

\begin{figure*}
    \centering
    \includegraphics[width=0.87\textwidth]{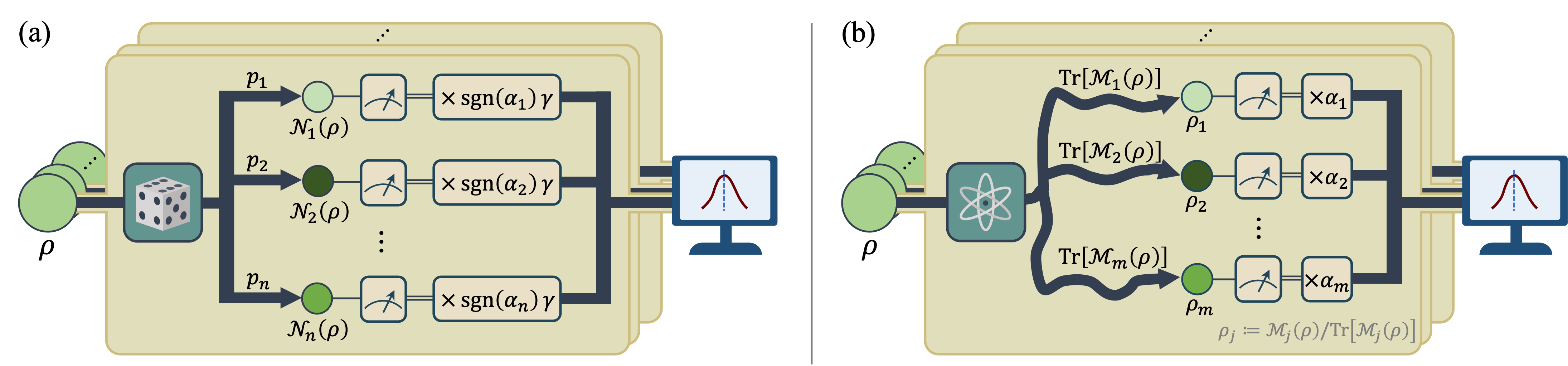}
    \caption{Difference between QPD and measurement-controlled post-processing. The task is to estimate the expectation value of an observable with respect to a state transformed by a Hermitian-preserving map without knowing the observable nor the state.
    (a) Estimating the expectation value with QPD. The action on the input state in each round is determined by classical random sampling, which is independent of the input state. (b) Estimating the expectation value with measurement-controlled post-processing. The action on the input state in each round is determined by quantum measurement governed by the Born rule, which takes the input state into consideration.}
    \label{fig:model_comparison}
\end{figure*}

\section{Measurement-controlled post-processing}
A Hermitian-preserving map, which maps Hermitian operators to Hermitian operators, can be written as a linear combination of CPTN maps. QPD simulates a Hermitian-preserving map using such a decomposition. For example, consider a decomposition of a Hermitian-preserving map $\cE = \sum_j \a_j\cN_j$, where each $\a_j$ is a real number and $\cN_j$ is a CPTN map.
To simplify the analysis, we assume that only the operation $\cE$ is applied. However, the analysis can be easily extended to cases in which $\cE$ is just one step in a sequence of operations as in, for instance, error mitigation, where noise needs to be canceled after each intermediate operation.
To begin with, QPD samples a CPTN map from $\lrc{\cN_j}$ with a probability distribution $p_j \coloneqq p(\cN_j) \coloneqq |\a_j| / \sum_j |\a_j|$.
Then, the sampled operation is applied to the input state $\rho$~\footnote{
As a sampled operation $\cN$ could be a CPTN map, it acting on the input state $\rho$ should give out a normalized state $\cN(\rho) / \tr[\cN(\rho)]$. However, a CPTN map is implemented by measurement and post-selection, and $\tr[\cN(\rho)]$ is exactly the success probability of implementing $\cN$. This probability cancels out the normalization factor in the output state. For this reason, when analyzing QPD, we can ignore the normalization factor $\tr[\cN(\rho)]$ as in Fig.~\ref{fig:model_comparison}(a).
} followed by a one-shot measurement with a given observable $O$.
When the sampled operation is $\cN_j$, a coefficient $\sgn(\a_j)\gamma$ will be multiplied to the outcome from the observable measurement, where $\gamma \coloneqq \sum_j |\a_j|$ and $\sgn$ denotes the sign function.
The post-processing coefficient ensures that the final output has an expected value $\sum_j p_j \sgn(\a_j)\gamma \tr[\cN_j(\rho) O] = \tr\lrb{\cE(\rho) O}$. After repeating the sampling, measurement, and post-processing for multiple rounds, we can get a fairly accurate estimation of $\tr\lrb{\cE(\rho) O}$ by taking the average of the outputs from all rounds. The whole process is shown in Fig.~\ref{fig:model_comparison}(a).

Compared with physical operations, unphysical maps could lead to a higher sampling overhead in terms of the number of one-shot measurements, or equivalently, the number of input state copies required to achieve the desired accuracy.
While the exact number of shots required could be hard to compute, Hoeffding's inequality offers a convenient upper bound, which has been used to characterize the cost of simulating unphysical maps in previous works~\cite{Jiang2020, Regula2021a}.
As we show in Appendix~\ref{app:overhead}, by Hoeffding's inequality~\cite{hoeffding1994probability}, the number of one-shot measurements $M$ satisfying
\begin{align}\label{eq:samp_over_qpd}
    M \geq \frac{2\gamma^2 \|O\|_\infty^2 \log\frac{2}{\delta}}{\varepsilon^2} = \gamma^2 K(\delta, \varepsilon, O)
\end{align}
are sufficient to ensure that the final estimation of $\tr\lrb{\cE(\rho) O}$ is within an error $\varepsilon$ with a probability no less than $1-\delta$, where $K(\delta, \varepsilon, O) \coloneqq 2\|O\|_\infty^2 \log\frac{2}{\delta} / \varepsilon^2$ is independent of the post-processing. The factor $\gamma = \sum_j |\a_j|$, which is the largest magnitude of post-processing coefficients, characterizes the sampling overhead of simulating $\cE$ with the decomposition $\cE = \sum_j \a_j\cN_j$ using QPD.

In the above model of QPD, which physical operation will be enacted in each round is resolved by a pre-determined classical probability distribution. As we are dealing with quantum information, there is no sense to limit the control system to a classical one.
What we need is an operation that takes in an input state and outputs a state for the succeeding observable measurement and a classical value controlling the post-processing.
The most general form of such a physical operation is known as a quantum instrument.

A quantum instrument is a quantum operation that gives both classical and quantum outputs. Mathematically, it is described by a collection $\lrc{\cM_j}$, where each $\cM_j$ is a CPTN map and $\sum_j \cM_j$ is CPTP~\cite{khatri2020principles}. Given an initial state $\rho$, the quantum instrument outputs a measurement outcome $j$ and a corresponding post-measurement state $\rho_j \coloneqq \cM_j(\rho) / \tr\lrb{\cM_j(\rho)}$ with a probability $\tr\lrb{\cM_j(\rho)}$ dependent on the input state.
When the measurement outcome from the quantum instrument is $j$, we multiply a coefficient $\a_j$ to the value obtained from the one-shot observable measurement on the state $\rho_j$. Then, the expected value of the output after classical post-processing is $\sum_j \tr\lrb{\cM_j(\rho)} \alpha_j \tr\lrb{\rho_j O} = \tr\lrb{\sum_j \alpha_j \cM_j(\rho) O}$. As in QPD, we repeat these steps for multiple rounds to obtain an estimation of $\tr\lrb{\sum_j \alpha_j \cM_j(\rho) O}$. We call this whole process measurement-controlled post-processing, which is visualized in Fig.~\ref{fig:model_comparison}(b).

In the example presented above, the map $\sum_j \alpha_j \cM_j$ is the effective operation performed on the input state.
While this decomposition looks a lot like the decomposition in QPD, the sampling overhead associated with it is $\a_{\max} \coloneqq \max_j |\a_j|$ instead of $\sum_j |\a_j|$.
This is because here we directly use the coefficients $\{\a_j\}$ in post-processing, and thus the largest magnitude among all the post-processing coefficients is $\a_{\max}$ instead of $\sum_j |\a_j|$ in QPD.
Specifically, it can be directly verified by Hoeffding's inequality that
\begin{align}\label{eq:samp_over_tc}
    M \geq \a_{\max}^2 K(\delta, \varepsilon, O)
\end{align}
one-shot measurements are required to make sure that the prediction has an error smaller than $\varepsilon$ with a probability no less than $1-\delta$.

The major distinction between QPD and measurement-controlled post-processing that leads to the difference in the expressions for sampling overheads is that QPD encodes a fixed probability distribution in the decomposition coefficients, whereas the latter does not impose such an artificial probability distribution on the measurement outcomes.
Instead, the probability distribution $\lrc{\tr\lrb{\cM_j(\rho)}}$ arises naturally from the the Born rule for the quantum measurement embedded in the CPTN maps constituting the quantum instrument, and the distribution varies from state to state.

\section{Twisted channel for mathematical characterization}
To assist the analysis of measurement-controlled post-processing, we introduce its mathematical characterization called a twisted channel.
The effective operation $\sum_j \alpha_j \cM_j$ can be written as $\alpha_{\max} \sum_j \a'_j \cM_j$, where the normalized coefficient $\a'_j \coloneqq \a_j / \alpha_{\max}$ satisfies $|\a'_j| \leq 1$.
Then, the combination $\sum_j \a'_j \cM_j$ has a unit sampling overhead and $\sum_j |\a'_j| \cM_j$ is a CPTN map. Without loss of generality, we can require $\sum_j |\a'_j| \cM_j$ to be CPTP. This is because there exists a map $\cM'$ such that $\sum_j |\a'_j| \cM_j + 2\cM'$ is CPTP. Then, we can add two terms $\cM'$ and $-\cM'$ into the combination $\sum_j \a'_j \cM_j$ without changing the map this combination represents for nor its sampling overhead.
Absorbing $|\a'_j|$ into each $\cM_j$, the effective operation can be written as $\a_{\max} \cT$, where $\cT$, though not necessarily physical, can be simulated with unit sampling overhead, and we call it a twisted channel~\footnote{
A set of operations similar to twisted channels has been noted in Ref.~\cite{piveteau2023circuit} inspired by Refs.~\cite{Mitarai2021, Mitarai2021a} for the task of circuit knitting. This set is introduced as a generalization of the set of CPTN maps to be sampled under the conventional QPD framework. In this work, we are taking an inherently different approach by involving quantum measurement in the process of generating probability distributions and sampling from them, where a twisted channel naturally arises.
}.

\begin{definition}[Twisted channel]\label{def:twisted_channel}
    A twisted channel $\cT$ is a linear map that can be written as $\cT = \sum_j s_j \cM_j$, where $s_j \in \lrc{+1, -1}$ and $\lrc{\cM_j}$ is a quantum instrument.
\end{definition}

It is clear that the effective operation of measurement-controlled post-processing is simply a single twisted channel scaled by a coefficient that coincides with the sampling overhead. As measurement-controlled post-processing is more general than the classically controlled one, one twisted channel with a suitable scalar is enough for simulating any Hermitian-preserving map, and we formally prove this statement in Appendix~\ref{app:one_instrument}. On the other hand, one may decompose a Hermitian-preserving map into a linear combination of multiple twisted channels as in QPD, which corresponds to sampling different quantum instruments at different measurement shots. It turns out that involving more quantum instruments does not result in a lower overhead.

\begin{theorem}[One quantum instrument is all you need]\label{theorem:one_quantum_instrument}
Under measurement-controlled post-processing, any protocol that involves the sampling of multiple quantum instruments is equivalent to a protocol using a single quantum instrument in terms of the simulated map and the sampling overhead.
\end{theorem}

The proof of Theorem~\ref{theorem:one_quantum_instrument} can be found in Appendix~\ref{app:one_instrument}, where we show that any linear combination of twisted channels can be reduced to a scaled twisted channel without changing the sampling overhead.
This theorem implies that measurement-controlled post-processing requires repeating only a single physical operation, which is more hardware-friendly compared to sampling from multiple physical operations as in QPD.

\section{Diamond norm as the simulation cost}
We have shown that in both QPD and measurement-controlled post-processing, a sampling overhead characterizing the required number of one-shot measurement rounds is associated with a decomposition of a Hermitian-preserving map. Considering that a Hermitian-preserving map's decomposition is not unique, we define a Hermitian-preserving map's simulation cost as the lowest possible sampling overhead associated with any its valid decomposition.
Within measurement-controlled post-processing, the simulation cost of a Hermitian-preserving map $\cE$ is defined as
\begin{align}\label{eq:cost_def}
    \gamma_\tc \lrp{\cE} \coloneqq \min \lrc{\a ~\middle|~ \cE = \a \cT,~ \a \geq 0,~ \cT \in \tc},
\end{align}
where $\tc$ denotes the set of twisted channels.
In the following, we demonstrate that this simulation cost is simply the diamond norm of the map.

The diamond norm, also known as the completely bounded trace norm, is a fundamental quantity in quantum information~\cite{kitaev1997quantum}. It serves as a measure of distance between quantum channels and is efficiently computable by a semi-definite program (SDP)~\cite{watrous2018theory}.
Given a Hermitian-preserving map $\cE_{A\to B}$, its diamond norm is defined as
\begin{align}
    \lrV{\cE}_\diamond \coloneqq \max_{\rho_{A'A}} \lrV{\id_{A'} \ox \cE_{A\to B}(\rho_{A'A})}_1,
\end{align}
where the optimization is over all states $\rho_{A'A}$ and the dimension of the system $A'$ is equal to the dimension of system $A$. The map $\id_{A'}$ denotes the identity channel on the system $A'$, and $\lrV{\cdot}_1$ is the trace norm.

Despite the widespread usage of diamond norm in quantum information theory, its operational interpretation for general Hermitian-preserving maps has been missing.
In the task of quantum channel discrimination, the diamond norm of the difference between a pair of quantum channels determines the maximum probability of successfully distinguishing between them in a one-shot setting~\cite{rosgen2005hardness, gilchrist2005distance}.
Recently, it was shown that the diamond norm of any Hermitian-preserving map that is proportional to a trace-preserving map is equal to the map's simulation cost using QPD~\cite{Regula2021a}.

However, for general Hermitian-preserving maps, this equality between the cost and the diamond norm does not hold. In an extreme case, the simulation cost can be twice as large as the map's diamond norm.
In Theorem~\ref{theorem:diamond_norm}, we show that this unpleasant gap between the cost in QPD and the diamond norm can be remedied. A formal proof of this theorem is given in Appendix~\ref{app:diamond}.

\begin{theorem}[Diamond norm is the cost]\label{theorem:diamond_norm}
    Let $\cE_{A\to B}$ be an arbitrary Hermitian-preserving map. Then, its simulation cost using a twisted channel can be obtained by the following SDP:
    \begin{align}\label{eq:cost_sdp}
        \gamma_\tc(\cE) =& \min \big\{ \alpha ~\big|~ J^{\cE} = M^+ - M^-,~ M^\pm \geq 0, \nonumber\\
        &\qquad\qquad \tr_B\lrb{M^+ + M^-} = \alpha \idop \big\},
    \end{align}
    where $J^{\cE}$ is the Choi operator of $\cE$.
    Furthermore, this cost is equal to the map's diamond norm, i.e.,
    \begin{align}\label{eq:cost_eq_diamond}
        \gamma_\tc(\cE) = \|\cE\|_\diamond.
    \end{align}
\end{theorem}
Eq.~\eqref{eq:cost_eq_diamond} can be readily obtained by observing that SDP~\eqref{eq:cost_sdp} is also an SDP for the diamond norm as shown in Ref.~\cite{Regula2021a}. This equality between $\gamma_\tc(\cE)$ and $\|\cE\|_\diamond$ not only implies that measurement-controlled post-processing can simulate a Hermitian-preserving map at a cost much lower than QPD, but also establishes the diamond norm as a universal quantity measuring the simulability of a Hermitian-preserving map.
Thus, Theorem~\ref{theorem:diamond_norm} assigns diamond norm its first well-defined operational meaning for every Hermitian-preserving map.

From another perspective, the cost of simulating an unphysical map characterizes the map's non-physicality~\cite{Regula2021a, Jiang2020, guo2023noise}. All conventional physical operations, i.e., CPTN maps, have unit simulation costs. Within the framework where classical post-processing is allowed, physical operations are extended to twisted channels. Intuitively, the more unphysical a map is, the more expensive it is to simulate this map with physical operations.
In this sense, we treat non-physicality as a resource and twisted channels as free operations. Then, a map's non-physicality can be quantified by robustness measures, which are widely used in quantum resource theories~\cite{chitambar2019quantum, vidal1999robustness, harrow2003robustness, steiner2003generalized, brandao2007entanglement, almeida2007noise, napoli2016robustness, skrzypczyk2019robustness}.

Here, we consider the absolute robustness~\cite{vidal1999robustness} of a Hermitian-preserving map $\cE$, which we define as
\begin{align}\label{eq:robustness_def}
    R(\cE) \coloneqq \min \lrc{\lambda ~\middle|~ \frac{\cE + \lambda \cT}{1 + \lambda} \in \tc,~ \cT \in \tc}.
\end{align}
It turns out that this robustness measure is equivalent to the diamond norm in the way that
\begin{align}
    \|\cE\|_\diamond = 2 R(\cE) + 1
\end{align}
for all Hermitian-preserving maps $\cE$.
A proof of this equality can be found in Appendix~\ref{app:diamond}.

\section{Advantages of measurement-controlled post-processing in practice}
Here, we study information recovering and processing as two examples to demonstrate the advantage of measurement-controlled post-processing over QPD in tasks of practical interest.

Information recovering~\cite{Zhao2022} refers to the task that predicts the expectation value of a given observable $O$ with respect to an arbitrary state $\rho$ given its noisy copies $\cN(\rho)$, where the channel $\cN$ represents the noise.
This is equivalent to optimizing a Hermitian-preserving map $\cD$ so that $\tr\lrb{\cD\circ\cN(\rho)O} = \tr\lrb{\rho O}$ for any state $\rho$ and a fixed observable $O$, where the Hermitian-preserving map can be simulated either by classically sampling CPTN maps using QPD or by measurement-controlled post-processing.
More details about this task can be found in Appendix~\ref{app:info_rec}.

We compare the lowest sampling overheads incurred by these two methods to recover the desired expectation value in an example, and the results are presented in Fig.~\ref{fig:info_recovery}. In this example, the given observable is the sum of the four Pauli operators $X+Y+Z+I$. We consider three different types of noise: the amplitude damping noise $\cN_{\rm AD}^\epsilon(\cdot)$ with Kraus operators $\proj{0} + \sqrt{1-\epsilon}\proj{1}$ and $\sqrt{\epsilon}\ketbra{0}{1}$, the dephasing noise $\cN_{\rm deph}^\epsilon$ with Kraus operators $\sqrt{1-\epsilon/2} I$ and $\sqrt{\epsilon/2} Z$, and the depolarizing noise $\cN_{\rm depo}^\epsilon(\cdot) = (1-\epsilon)(\cdot) + \epsilon\tr[\cdot]I/2$. The parameter $\epsilon\in[0,1]$ indicates the noise level for all the three channels.
The lowest sampling overheads are solved using SDPs provided in Ref.~\cite{Zhao2022}. For measurement-controlled post-processing, this is equivalent to minimizing the diamond norm over all feasible maps $\cD$.
It is observed in Fig.~\ref{fig:info_recovery} that for all these noises, the twisted channel method incurs overheads significantly lower than those of QPD, and the gaps between them steadily enlarge as the noise level goes up.
In Appendix~\ref{app:info_rec}, we also present numerical results verifying that the number of measurement shots corresponding to the lowest sampling overhead incurred by measurement-controlled post-processing is sufficient to achieve an desired level of accuracy.

\begin{figure}[t]
    \centering
    \includegraphics[width=0.82\columnwidth]{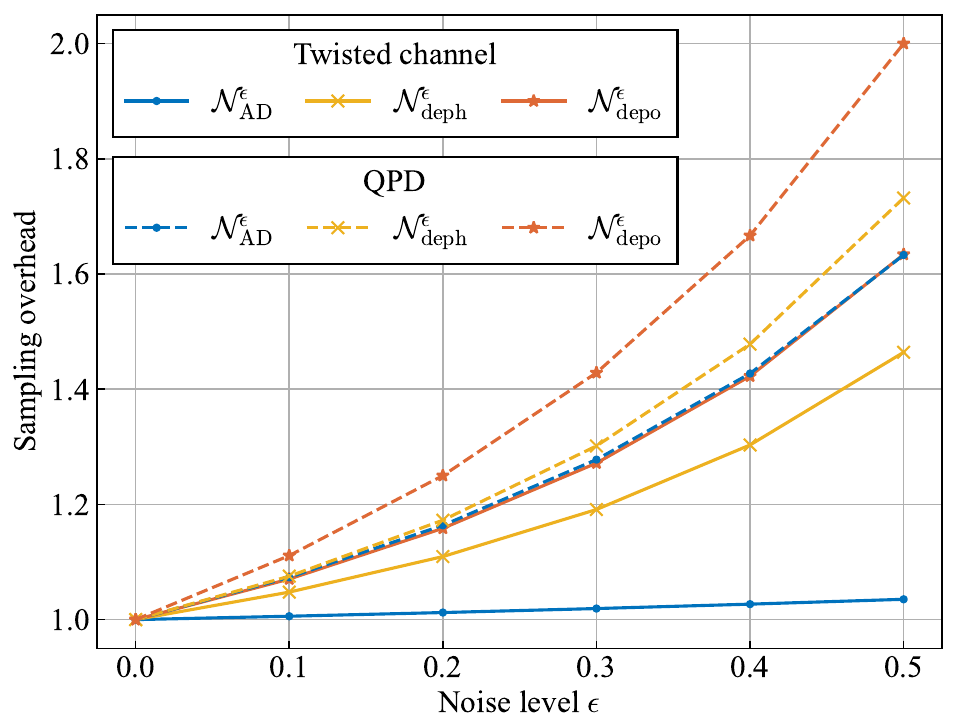}
    \caption{Comparison between sampling overheads of QPD and measurement-controlled post-processing for information recovering under common noises at different noise levels. The markers on the solid lines represent the overheads achieved by using twisted channels, i.e., measurement-controlled post-processing, and those on the dashed lines are the overheads achieved by QPD.}
    \label{fig:info_recovery}
\end{figure}

In addition to information recovering, the other application we consider here is the processing of quantum data, which involves a collection of tasks aimed at implementing extraction maps on input quantum states. Quantum algorithms have the potential to achieve quantum speed-up owing to the vast information storage capabilities of quantum systems. However, this advantage also poses challenges in processing valuable information when quantum data contain a surplus of irrelevant details. To optimize the efficiency and accuracy of quantum algorithms, it is crucial to implement an extraction map that minimizes the spatial dimensions of quantum data while retaining as much useful information as possible. This is especially important for some quantum machine learning tasks. For example, in the quantum convolutional neural networks proposed in Ref.~\cite{cong2019quantum}, such operations are implemented by measurement-controlled unitaries and hence are restricted to physical maps only.

\begin{figure*}[t]
    \centering
    \includegraphics[width=0.9\textwidth]{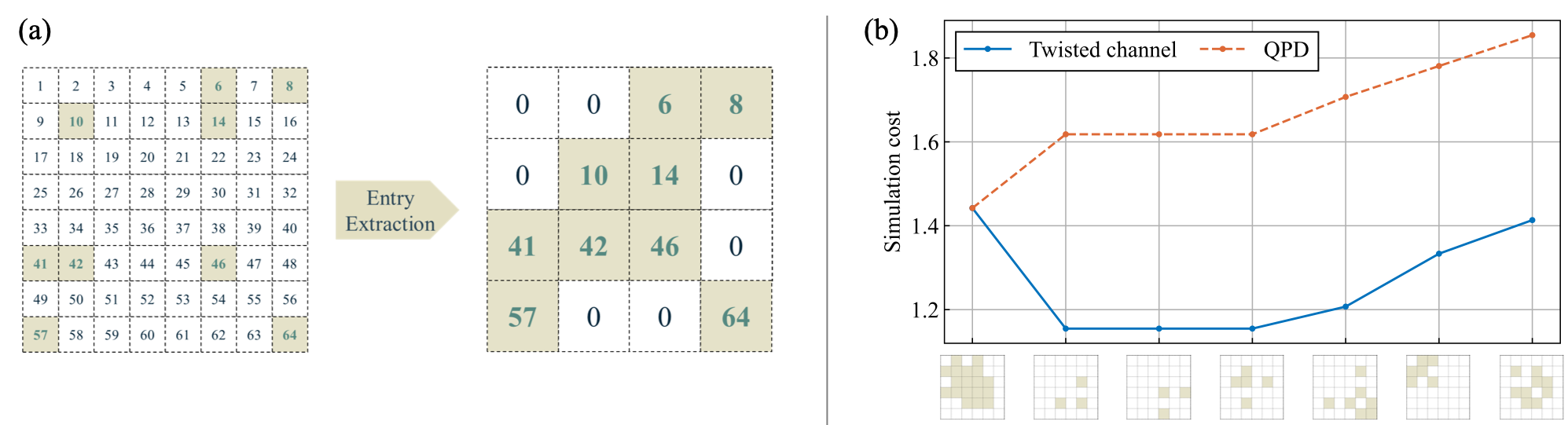}
    \caption{Illustration of entry extraction maps and a comparative evaluation of their implementation using QPD and the twisted channel method. (a) A typical instance of an entry extraction map, which converts an $8 \times 8$ matrix into a $4 \times 4$ matrix that comprises all the extracted entries while preserving their relative locations. (b) A plot comparing the costs of QPD and the twisted channel method across various entry extraction maps. All the maps share a common input dimension and are distinguished by the index sets they extract, as indicated on the horizontal axis.}
    \label{fig:entry extraction}
\end{figure*}

Twisted channels can expand the range of applicable extraction to encompass unphysical maps. We consider an elementary operation called the \emph{entry extraction map} as an example. In particular, entry extraction maps extract entry-wise information of an input quantum state and blend these entries with zeros to form a new matrix while preserving the relative positions of the entries. Additionally, when an off-diagonal entry is extracted, its symmetrical counterpart must be extracted as well to maintain Hermiticity. The graphical illustration of an entry extraction map can be found in Fig.~\ref{fig:entry extraction}(a), and its mathematical definition is presented in Appendix~\ref{app:entry_ext}.

To showcase the efficiency of quantum measurement, we compare the simulation costs between QPD and the twisted channel method for several entry extraction maps. The input quantum data for these maps are fixed to $6 \times 6$ matrices, while the output dimensions and indices sets extracted by these map are randomly selected. Fig.~\ref{fig:entry extraction}(b) illustrates that for all the selected extracting operations, measurement-controlled post-processing achieves the same or lower costs than QPD.

\section{Quantum interference supplies the advantage}
The advantage of measurement-controlled post-processing over QPD is substantiated by the numerical results presented above. Consequently, it is pertinent to inquire about the underlying physical property that contributes to these enhancements. Here, we suggest that quantum interference is the key of such improvements.

Upon closer examination, the only difference between measurement-controlled post-processing and QPD lies in their respective methodologies of operation sampling.
QPD first samples an operation according to a fixed priori probability distribution and then performs the sampled operation. On the other hand, measurement-controlled post-processing does not have such a priori probability distribution.
As depicted in Fig.~\ref{fig:model_comparison}(b), it first applies an aggregate quantum operation to the input state, which creates interference between the paths leading to different measurement outcomes associated with different physical operations waiting to be sampled.
The probability of getting each outcome is affected by the interference depending on the input state.
Upon measurement, the quantum system collapses to an output state corresponding to one particular operation based on the adaptive probability distribution.
Such dynamic assignment of probabilities makes interference the core ingredient for the enhancements brought by measurement-controlled post-processing.

\section{Discussion}
We demonstrate the power of quantum measurement in an important practical task by showing that it leads to a lower simulation cost of unphysical operations compared with QPD, a standard method widely used for unphysical operation simulation. The simulation costs when quantum measurement is employed reduce to the well-known diamond norm, thereby endowing the diamond norm with the first operational interpretation applicable to all Hermitian-preserving maps.

The measurement-controlled post-processing scheme introduced in this study has potential applications beyond the simulation of unphysical maps.
As more quantum algorithms and protocols incorporate classical randomness and post-processing as valuable tools, it is imperative to explore whether quantum measurement can be utilized to improve performance in various scenarios.
Potential use cases include circuit knitting~\cite{Mitarai2021a,piveteau2023circuit,jing2025circuit}, Hamiltonian simulation~\cite{campbell2019random, faehrmann2022randomizing, kiss2023importance}, and quantum error correction~\cite{piveteau2021error, lostaglio2021error, suzuki2022quantum}.
Investigating the feasibility and utility of employing quantum measurement in these scenarios represents an exciting avenue for future work.
Furthermore, the demonstrated advantage of quantum measurement in the task considered in this work opens up new possibilities for identifying practical computational advantages that can be realized in the near term. Our findings could serve as a promising starting point for further research in this direction.

\section*{Acknowledgments}
The authors thank Mingrui Jing and Wenjun Yu for fruitful discussions.
This work was partially supported by the National Key R\&D Program of China (Grant No.~2024YFB4504004, 2024YFE0102500), the National Natural Science Foundation of China (Grant. No.~12447107), the Guangdong Provincial Quantum Science Strategic Initiative (Grant No.~GDZX2403008, GDZX2403001), the Guangdong Provincial Key Lab of Integrated Communication, Sensing and Computation for Ubiquitous Internet of Things (Grant No.~2023B1212010007), the Quantum Science Center of Guangdong-Hong Kong-Macao Greater Bay Area, and the Education Bureau of Guangzhou Municipality.
%

\appendix
\section{Sampling Overhead}\label{app:overhead}
Here, we rigorously show that the sampling overhead of either QPD or measurement-controlled post-processing can be characterized by the largest magnitude of post-processing coefficients.
Hoeffding's inequality~\cite{hoeffding1994probability} states that, given $M$ random variables $X_1,\dots,X_M$, if $X_j$ take its value from intervals $[a_j, b_j]$ for every $j$ from $1$ to $M$, then the sum of these random variables $S_M = \sum_{j=1}^M X_j$ satisfies
\begin{align}
    \Pr(|S_M - {\rm E}[S_M]| \geq \varepsilon) \leq 2\exp\lrp{-\frac{2\varepsilon^2}{\sum_{j=1}^M (b_i - a_i)^2}}
\end{align}
for any error tolerance $\varepsilon \geq 0$, where ${\rm E}[S_M]$ denotes the expected value of $S_M$.

In our case, no matter we use measurement-controlled post-processing or QPD, in each round, we perform a one-shot measurement with a given observable $O$ to get a measurement outcome $o$ and then multiply this outcome with some post-processing coefficient $\lambda$. After $M$ rounds of sampling and measurement, the final estimation of the expectation value of $O$ is given by
\begin{align}
    S_M = \frac{1}{M} \sum_{j=1}^M \lambda_j o_j,
\end{align}
where $\lambda_j$ and $o_j$ denote the post-processing coefficient and the measurement outcome, respectively, in the $j$-th round. Applying Hoeffding's inequality to our case, each random variable $X_j$ corresponds to the random variable $\lambda_j o_j / M$.
The measurement outcome $o_j$ is a random variable such that $-\|O\|_\infty \leq o_j \leq \|O\|_\infty$.
As a result, we have $-|\lambda_j|\|O\|_\infty / M \leq X_j \leq |\lambda_j|\|O\|_\infty / M$.
The post-processing coefficient $\lambda_j$ is also a random variable whose value depends on which CPTN map is sampled in this round, and $|\lambda_j|$ is bounded as $0 \leq |\lambda_j| \leq \lambda_{\rm max}$, where $\lambda_{\rm max}$ is the largest absolute value among all post-processing coefficients. Therefore, we obtain $-\lambda_{\rm max}\|O\|_\infty / M \leq X_j \leq \lambda_{\rm max}\|O\|_\infty / M$.
Then, by Hoeffding's inequality, we arrive at
\begin{align}
    &\Pr(|S_M - {\rm E}[S_M]| \geq \varepsilon) \nonumber\\
    \leq &2\exp\lrp{-\frac{2\varepsilon^2}{\sum_{j=1}^M \lrp{\frac{2}{M}\lambda_{\rm max}\|O\|_\infty}^2}}\\
    = &2\exp\lrp{-\frac{M \varepsilon^2}{2 \lambda_{\rm max}^2 \|O\|_\infty^2 }}.
\end{align}
Note that ${\rm E}[S_M]$ here is the ideal expectation value that we aim to estimate.
If we want the probability of $|S_M - {\rm E}[S_M]| \geq \varepsilon$ to be no larger than a preset value $\delta$, it would be sufficient to require that
\begin{align}
    2\exp\lrp{-\frac{M\varepsilon^2}{2 \lambda_{\rm max}^2 \|O\|_\infty^2}} \leq \delta
\end{align}
With straightforward algebraic calculation, one can find that the requirement above can be equivalently expressed as a requirement on the number of sampling rounds $M$:
\begin{align}\label{appeq:hoeffding}
    M \geq \frac{2 \lambda_{\rm max}^2 \|O\|_\infty^2 \log\frac{2}{\delta}}{\varepsilon^2}.
\end{align}
As $\|O\|_\infty$, $\epsilon$, and $\delta$ are independent of what sampling method we use, this inequality implies that $\lambda_{\rm max}$ is indeed the key quantity characterizing the performance of a sampling method.

For QPD, the post-processing coefficient $\lambda_j$ in each round is either $\gamma$ or $-\gamma$, and thus $\lambda_{\rm max} = \gamma$, reducing Eq.~\eqref{appeq:hoeffding} to Eq.~\eqref{eq:samp_over_qpd}.
For measurement-controlled post-processing, $\lambda_{\rm max}$ is simply $\a_{\rm max}$, reducing Eq.~\eqref{appeq:hoeffding} to Eq.~\eqref{eq:samp_over_tc}.

\section{One Quantum Instrument Is All You Need}\label{app:one_instrument}
Before we give a proof of Theorem~\ref{theorem:one_quantum_instrument}, we first show that a single twisted channel scaled by a suitable coefficient is enough to simulate any Hermitian-preserving map.
\begin{proposition}
    A linear map $\cE$ can be written as $\a \cT$ for a twisted channel $\cT$ and a real number $\a$ if and only if it is Hermitian-preserving.
\end{proposition}
\begin{proof}
    For the ``only if'' part, we note that $\cT$ is Hermitian-preserving as it is a linear combination of CPTN maps, which are Hermitian-preserving. Hence, $\a\cT$ is a Hermitian-preserving map.

    For the ``if'' part, we note that any Hermitian-preserving linear map $\cE$ can be written as a linear combination of CPTN maps that constitute a quantum instrument, i.e., $\cE = \sum_j \a_j\cN_j$ for each $\cN_j$ being CPTN and $\sum_j \cN_j$ being CPTP (see Proposition 1 in Ref.~\cite{buscemi2013direct}). Denoting $\a_{\max} \coloneqq \max_j |\a_j|$ and $\a_j' \coloneqq \a_j / \a_{\max}$, we have $\cE = \a_{\max} \sum_j \a_j' \cN_j$. By definition, $|\a_j'| \leq 1$ and thus each $\cM_j \coloneqq |\a_j'| \cN_j$ is CPTN and so is $\sum_j \cM_j$. Let $\cM'$ be a CPTN map such that $\sum_j \cM_j + 2\cM'$ is CPTP. Note that such a map always exists. Then, we can write the Hermitian-preserving map $\cE$ as $\cE = \a_{\max} \lrp{\sum_j s_j \cM_j + \cM' - \cM'}$, where $s_j \coloneqq \sgn(\a_j')$ denotes the sign of $\a_j'$ and $\sum_j s_j \cM_j + \cM' - \cM'$ is a twisted channel.
\end{proof}

Now we show that a single twisted channel is not only enough, but also optimal. That is, as stated in Theorem~\ref{theorem:one_quantum_instrument}, allowing the sampling of multiple quantum instruments does not lower the sampling overhead of simulation.

\renewcommand\theproposition{\ref{theorem:one_quantum_instrument}}
\setcounter{proposition}{\arabic{proposition}-1}
\begin{theorem}
    Under measurement-controlled post-processing, any protocol that involves the sampling of multiple quantum instruments is equivalent to a protocol using a single quantum instrument in terms of the simulated map and the sampling overhead.
\end{theorem}
\renewcommand{\theproposition}{\arabic{proposition}}

A measurement-controlled post-processing protocol with multiple quantum instruments is described as a linear combination of twisted channels and its sampling overhead is the sum of the absolute values of the coefficients in the combination as in QPD. In the following lemma, we show that any linear combination of twisted channels can be rewritten as a single twisted channel scaled by the combination's sampling overhead.

\begin{lemma}\label{lemma:one_twisted_channel}
    For any linear combination of twisted channels $\sum_j \a_j \cT_j$, there exists a single twisted channel $\cT$ such that $\a \cT = \sum_j \a_j \cT_j$ for $\a = \sum_j |\a_j|$.
\end{lemma}
\begin{proof}
    Without loss of generality, we can assume that every $\a_j$ is positive because we can flip the sign of the corresponding twisted channel otherwise. Moreover, any twisted channel $\cT = \sum_j s_j \cM_j$ can be written as a difference between two CPTN maps that add up to a CPTP map, i.e., $\cT = \cT^+ - \cT^-$, where $\cT^\pm \coloneqq \sum_{s_j = \pm1} \cM_j$. Then, denoting $\sum_j \a_j$ by $\a$ and $\a_j/\a$ by $p_j$, the linear combination can be simplified as
    \begin{align}
        \sum_j \a_j \cT_j = \a \sum_j p_j (\cT_j^+ - \cT_j^-) = \a \lrp{\cW^+ - \cW^-}
    \end{align}
    where $\cW^\pm \coloneqq \sum_j p_j \cT_j^\pm$. Note that $\sum_j p_j = 1$, and $\cT_j^\pm$ are CPTN, so $\cW^\pm$ are also CPTN due to the convexity of the set of CPTN maps. Also, $\cW^+ + \cW^- = \sum_j p_j (\cT_j^+ + \cT_j^-)$ is CPTP because $\cT_j^+ + \cT_j^-$ is CPTP for each $j$ and the set of CPTP maps is convex. That is, $\cW^+ - \cW^-$ is a twisted channel.
\end{proof}

A scaled twisted channel $\a\cT$ has a sampling overhead $|\a|$. On the other hand, a linear combination of twisted channels $\sum_j \a_j \cT_j$ has a sampling overhead of $\sum_j |\a_j|$ as we need to classically sample the quantum instrument corresponding to each $\cT_j$ as in QPD. Since Lemma~\ref{lemma:one_twisted_channel} shows that for any such linear combination, there is always a scaled twisted channel $\a\cT$ satisfying $\a\cT = \sum_j \a_j \cT_j$ and $\a =\sum_j |\a_j|$, we conclude that there exists a single quantum instrument (corresponds to $\cT$) that simulates the same Hermitian-preserving map as sampling a set of quantum instruments (corresponds to $\lrc{\cT_j}$) at the same overhead under measurement-controlled post-processing.

\section{Diamond Norm Is the Cost}\label{app:diamond}
The diamond norm of a Hermitian-preserving map $\cE_{A\to B}$ is defined as
\begin{align}
    \lrV{\cE}_\diamond \coloneqq \max_{\rho_{A'A}} \lrV{\id_{A'} \ox \cE_{A\to B}(\rho_{A'A})}_1,
\end{align}
where the optimization is over all states $\rho_{A'A}$ and the dimension of the system $A'$ is equal to the dimension of system $A$. The map $\id_{A'}$ denotes the identity channel on the system $A'$, and $\lrV{\cdot}_1$ is the trace norm.
Here, we derive the SDP given in Theorem~\ref{theorem:diamond_norm} for the cost of simulating any Hermitian-preserving map using a twisted channel, which happens to be an SDP for the map's diamond norm as well.
\renewcommand\theproposition{\ref{theorem:diamond_norm}}
\setcounter{proposition}{\arabic{proposition}-1}
\begin{theorem}
    Let $\cE_{A\to B}$ be an arbitrary Hermitian-preserving map. Then, its simulation cost using a twisted channel can be obtained by the following SDP:
    \begin{align}
    \begin{aligned}
        \gamma_\tc(\cE) = \min \big\{ \alpha ~\big|~ &J^{\cE} = M^+ - M^-,~ M^\pm \geq 0,\\
        &\tr_B\lrb{M^+ + M^-} = \alpha \idop \big\},
    \end{aligned}
    \end{align}
    where $J^{\cE}$ is the Choi operator of $\cE$.
    Furthermore, this cost is equal to the map's diamond norm, i.e.,
    \begin{align}
        \gamma_\tc(\cE) = \|\cE\|_\diamond.
    \end{align}
\end{theorem}
\renewcommand{\theproposition}{\arabic{proposition}}
\begin{proof}
The simulation cost of a Hermitian-preserving map $\cE$ with a twisted channel is the smallest possible $\a$ such that $\a\cT = \cE$ for a twisted channel $\cT$ and $\a \geq 0$.
In the proof of Lemma~\ref{lemma:one_twisted_channel}, we showed that any twisted channel can be represented as a difference between two CPTN maps, which leads to the following formulation of the simulation cost:
\begin{align}\label{appeq:tc_cost_original}
\begin{aligned}
    \gamma_\tc \lrp{\cE} = \min \big\{\a ~\big|~ &\cE = \a (\cM^+ - \cM^-),~ \a \geq 0,\\
    &\cM^\pm~\text{are CPTN},~ \cM^+ + \cM^-~\text{is CPTP}\big\}.
\end{aligned}
\end{align}
By the \Choi isomorphism~\cite{jamiolkowski1972linear, choi1975completely}, we can write the minimization above using the Choi operators of $\cE$ and $\cM^{\pm}$. A linear map $\cN_{A\to B}$'s Choi operator is defined as
\begin{align}
    J^\cN_{AB} \coloneqq \id_{A} \ox \cN_{A'\to B} (\proj{\Gamma}_{AA'}),
\end{align}
where $\ket{\Gamma}_{AA'} \coloneqq \sum_{j=0}^{d_A} \ket{j}_A\ket{j}_{A'}$ is the unnormalized maximally entangled state, $d_A$ is the dimension of the Hilbert space $\cH_A$ associated to system $A$, and $\cH_{A'}$ is isomorphic to $\cH_A$. The linear map $\cN_{A\to B}$ is completely positive if and only if its Choi operator is positive. It is trace-preserving if and only if $\tr_B\lrb{J^\cN_{AB}} = \idop_A$, and it is trace-non-increasing if and only if $\tr_B\lrb{J^\cN_{AB}} \leq \idop_A$. Then, Eq.~\eqref{appeq:tc_cost_original} can be rewritten as
\begin{align}
\begin{aligned}
    \gamma_\tc \lrp{\cE} = \min \Bigl\{\a ~\Big|~ &J^{\cE} = \a \lrp{J^{\cM^+} - J^{\cM^-}},~ \a \geq 0,\\
    &J^{\cM^\pm} \geq 0,~ \tr_B\lrb{J^{\cM^\pm}} \leq \idop_A,\\
    &\tr_B\lrb{J^{\cM^+} + J^{\cM^-}} = \idop_A\Bigr\}.
\end{aligned}
\end{align}
Denoting $M^\pm \coloneqq \a J^{\cM^\pm}$ and take them into the equation above, we obtain the following SDP:
\begin{align}
\begin{aligned}
    \gamma_\tc \lrp{\cE} = \min \bigl\{\a ~\big|~ &J^{\cE} = M^+ - M^-,~ \a \geq 0,\\
    &M^\pm \geq 0,~ \tr_B\lrb{M^\pm} \leq \a\idop_A,\\
    &\tr_B\lrb{M^+ + M^-} = \a\idop_A\bigr\}.
\end{aligned}
\end{align}
Note that both $\a \geq 0$ and $\tr_B\lrb{M^\pm} \leq \a\idop_A$ are already implied by $M^\pm \geq 0$ and $\tr_B\lrb{M^+ + M^-} = \a\idop_A$.
Thus, we can trim this SDP to arrive at the one stated in Theorem~\ref{theorem:diamond_norm}:
\begin{align}\label{appea:tc_cost_final}
\begin{aligned}
    \gamma_\tc \lrp{\cE} = \min \big\{\a ~\big|~ &J^{\cE} = M^+ - M^-,~ M^\pm \geq 0,\\
    &\tr_B\lrb{M^+ + M^-} = \a\idop_A\big\}.
\end{aligned}
\end{align}

For the equivalence between $\gamma_\tc$ and the diamond norm, we refer to Ref.~\cite{Regula2021a}, where the authors show that the SDP in Eq.~\eqref{appea:tc_cost_final} is an SDP for the diamond norm. This completes the proof of Theorem~\ref{theorem:diamond_norm}.
\end{proof}

\begin{remark}
    According to its definition in Eq.~\eqref{eq:cost_def}, the simulation cost $\gamma_\tc$ is the gauge~\cite{aubrun2017alice} of the set of twisted channels. Hence, Theorem~\ref{theorem:diamond_norm} implies that the gauge of the set of twisted channels is simply the diamond norm, or, equivalently, the set of twisted channels is the unit ball in the vector space of Hermitian-preserving maps under the diamond norm.
\end{remark}

In the main text, we also define a related quantity, that is, the absolute robustness $R(\cE)$ of a Hermitian-preserving map $\cE$:
\begin{align}\label{appeq:robustness_def}
    R(\cE) \coloneqq \min \lrc{\lambda ~\middle|~ \frac{\cE + \lambda \cT}{1 + \lambda} \in \tc,~ \cT \in \tc}.
\end{align}
Below, we show that there is a constant relation between the diamond norm and the absolute robustness. This establishes the equivalence between these two quantities over all Hermitian-preserving maps.

\begin{proposition}\label{prop:robustness_and_diamond}
    For any Hermitian-preserving map $\cE$, it holds that
    \begin{align}
        R(\cE) = \frac{\lrV{\cE}_\diamond - 1}{2}.
    \end{align}
\end{proposition}
\begin{proof}
    Let $\Phi$ be a twisted channel such that $\Psi \coloneqq (\cE + R(\cE)\Phi)/(1+R(\cE)) \in \tc$. Then we can write $\cE$ as $\cE = (1+R(\cE)) \Psi - R(\cE)\Phi$, which is a valid decomposition of $\cE$ into twisted channels. The sampling overhead associated with this decomposition is $1 + 2R(\cE)$. Since $\lrV{\cE}_\diamond = \gamma_\tc(\cE)$ is the minimized sampling overhead over all possible decompositions, we arrive at $\lrV{\cE}_\diamond \leq 1 + 2R(\cE)$, or, equivalently, $R(\cE) \geq (\lrV{\cE}_\diamond - 1) / 2$.

    On the other hand, suppose $\cT$ is a twisted channel that simulates $\cE$ with the optimal sampling overhead $\lrV{\cE}_\diamond$, i.e., $\cE = \lrV{\cE}_\diamond \cT$.
    Then, we can write $\cT$ as
    \begin{align}\label{appeq:feasible_robustness}
        \cT = \frac{\cE + s(-\cT)}{1 + s},
    \end{align}
    where $s \coloneqq \lrp{\lrV{\cE}_\diamond - 1} / 2$. As both $\cT$ and $-\cT$ are twisted channels, Eq.~\eqref{appeq:feasible_robustness} is a feasible solution to the minimization problem that defines the robustness in Eq.~\eqref{appeq:robustness_def}. Thus, we have $R(\cE) \leq s = (\lrV{\cE}_\diamond - 1) / 2$.

    Now we have established that $\lrp{\lrV{\cE}_\diamond - 1} / 2 \leq R(\cE) \leq \lrp{\lrV{\cE}_\diamond - 1} / 2$.
    Therefore, we conclude that $R(\cE) = (\lrV{\cE}_\diamond - 1) / 2$ for any Hermitian-preserving map $\cE$.
\end{proof}

\section{Information Recovering}\label{app:info_rec}
The task of information recovering \cite{Zhao2022} aims to extract an unbiased estimation of the expectation value of a given observable $O$ from multiple copies of a noisy state $\cN(\rho)$, where $\cN$ is a noise channel and $\rho$ is an arbitrary quantum state. In other words, the task aims to estimate the expectation value $\tr[\rho O]$ from the noisy state $\cN(\rho)$ for arbitrary $\rho$, with fixed observable $O$ and fixed noise $\cN$.
The idea used in Ref.~\cite{Zhao2022} was to find a Hermitian-preserving map $\cD$ such that
\begin{equation}\label{eq:app_info_recover}
    \cN^\dagger\circ\cD^\dagger(O) = O.
\end{equation}
Then, the target value can be estimated by applying the Hermitian-preserving map $\cD$ to the noisy state $\cN(\rho)$ and measuring the observable $O$, resulting in
\begin{equation}
    \tr[\cD\circ\cN(\rho)O] = \tr[\rho \cN^\dagger\circ\cD^\dagger(O)] = \tr[\rho O],
\end{equation}
where the second equality follows from Eq.~\eqref{eq:app_info_recover}.

Note that in this task, the goal is not to simulate a particular unphysical map, but to search all unphysical maps satisfying Eq.~\eqref{eq:app_info_recover}, which we call feasible maps, for one with the lowest sampling overhead.
In Ref.~\cite{Zhao2022}, a feasible map $\cD$ is considered to be realized by QPD, which decomposes $\cD$ into a linear combination of two CPTN maps $\cD_1$ and $\cD_2$: $\cD = c_1\cD_1 + c_2\cD_2$. The optimal sampling overhead of information recovering using QPD is given by
\begin{align}
\begin{aligned}
    \gamma^*_O(\cN) \coloneqq \min &\big\{|c_1| + |c_2| \,\big|\, \cD = c_1\cD_1 + c_2\cD_2,\\
    &\qquad\cD_1,\cD_2\in {\rm CPTN}, \; \cN^\dagger\circ\cD^\dagger(O) = O\big\}.
\end{aligned}
\end{align}
This optimal sampling overhead is defined as the lowest simulation cost of any feasible map and can be calculated by an SDP using the \Choi isomorphism.

In the present work, we consider simulating a feasible map $\cD$ by measurement-controlled post-processing in place of QPD. Accordingly, the optimal sampling overhead is given by
\begin{align}\label{appeq:tc_opt_overhead}
\begin{aligned}
    \tau^*_O(\cN) = \min \big\{\gamma_\tc \lrp{\cD} \,\big|\, &\cD\,\text{is Hermitian-preserving},\\
    &\cN^\dagger\circ\cD^\dagger(O) = O\big\},
\end{aligned}
\end{align}
which can also be computed by an SDP using the \Choi isomorphism.

In order to showcase the advantage of measurement-controlled post-processing, we compare the optimal sampling overhead of QPD and measurement-controlled post-processing under different noise channels. We fix the observable as $O=I+X+Y+Z$ and consider three single-qubit noise models: the amplitude damping noise $\cN_{\rm AD}^\epsilon(\cdot)$ with Kraus operators $\proj{0} + \sqrt{1-\epsilon}\proj{1}$ and $\sqrt{\epsilon}\ketbra{0}{1}$, the dephasing noise $\cN_{\rm deph}^\epsilon$ with Kraus operators $\sqrt{1-\epsilon/2} I$ and $\sqrt{\epsilon/2} Z$, and the depolarizing noise $\cN_{\rm depo}^\epsilon(\cdot) = (1-\epsilon)(\cdot) + \epsilon\tr[\cdot]I/2$. The parameter $\epsilon\in[0,1]$ indicates the noise level for all the three channels.
It can be seen in Fig.~\ref{fig:info_recovery} that the optimal sampling overhead of measurement-controlled post-processing is consistently lower than that of QPD.
Since the set of all feasible maps is the same regardless of what simulation method is used, Fig.~\ref{fig:info_recovery} demonstrates the sample-efficiency of measurement-controlled post-processing in simulating unphysical maps in a practical task.

Besides the numerical results on the optimal sampling overhead, we also conduct a numerical experiment to show that the number of observable measurement shots corresponding to the optimal sampling overhead is indeed sufficient to guarantee the desired accuracy.
In this experiment, we focus on the depolarizing noise $\cN_{\rm depo}^\epsilon$ and denote by $\cD^\epsilon$ the Hermitian-preserving map that achieves the optimal sampling overhead $\tau^*_O(\cN_{\rm depo}^\epsilon)$ as in Eq.~\eqref{appeq:tc_opt_overhead}. Then, according to Theorem~\ref{theorem:diamond_norm}, we know that
\begin{align}
    \tau^*_O(\cN_{\rm depo}^\epsilon) = \gamma_\tc \lrp{\cD^\epsilon} = \lrV{\cD^\epsilon}_\diamond,
\end{align}
implying that $\lrV{\cD^\epsilon}_\diamond^2 K(\delta, \varepsilon, O)$ measurement shots is sufficient to achieve an estimation within an error $\varepsilon$ with a probability no less than $1-\delta$.
To verify this, we generate a random input state $\rho$ and prepare its noisy version $\cN_{\rm depo}^\epsilon(\rho)$. Then we simulate the action of $\cD^\epsilon$ on the noisy state $\cN_{\rm depo}^\epsilon(\rho)$ and measure the resulted state with the observable $O=I+X+Y+Z$ to estimate the expectation value $\tr[\rho O]$. The numerical results are plotted in Fig.~\ref{fig:sample_to_estimate}. It is clear to see that as the number of measurement shots increases, the estimation becomes more likely to be within the error tolerance. When the number of measurement shots is $\lrV{\cD^\epsilon}_\diamond^2 K(\delta, \varepsilon, O)$, all the estimations are with in the error tolerance. That is, the success probability is $1$, which is clearly larger than $1-\delta = 0.9$.

\begin{figure}[h]
    \centering
    \includegraphics[width=0.82\columnwidth]{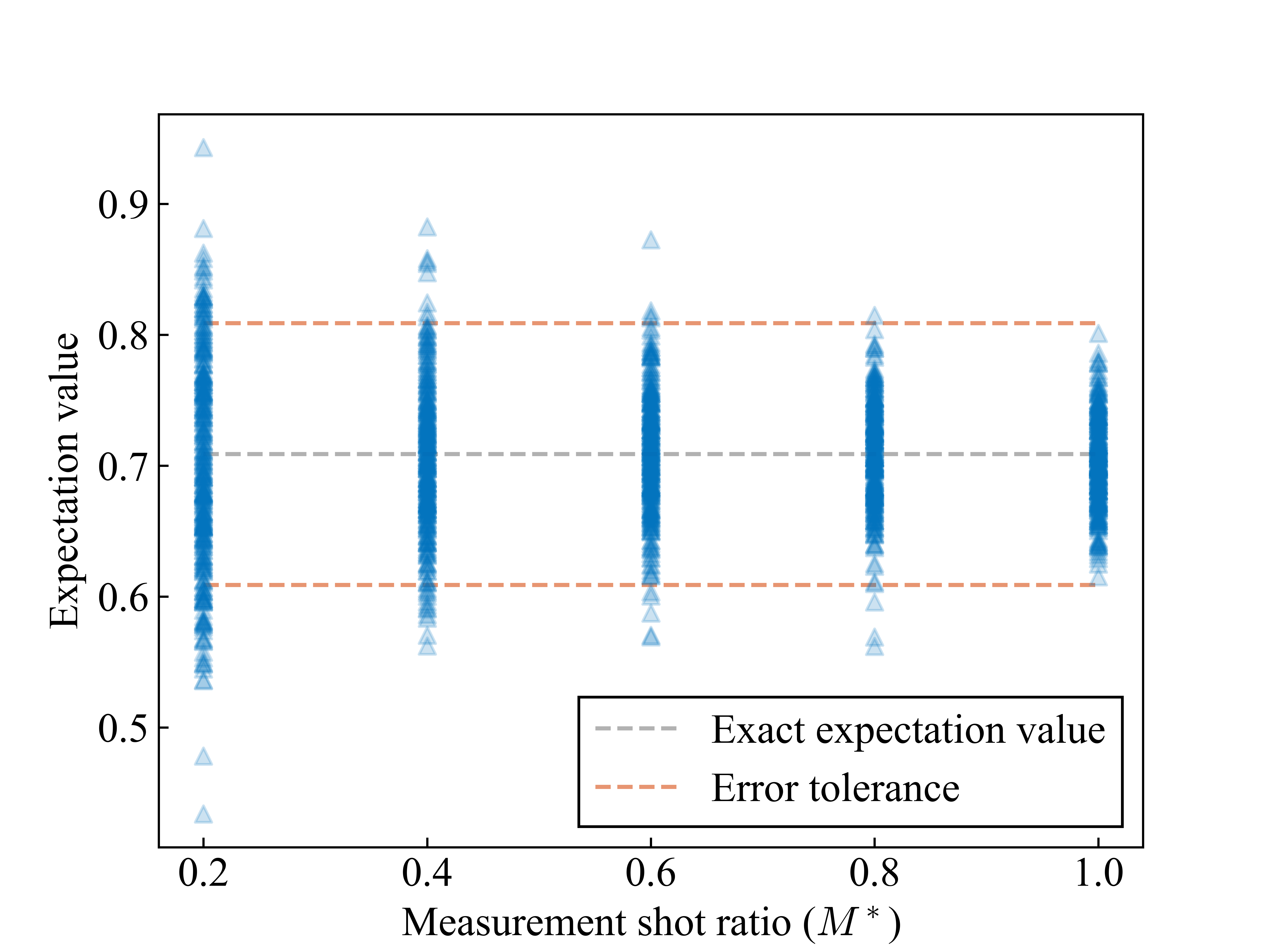}
    \caption{Estimation of the expectation value with different numbers of measurement shots. The number of measurement shots is measured in terms of its ratio to $M^* \coloneqq \lrV{\cD^\epsilon}_\diamond^2 K(\delta, \varepsilon, O)$, where we set $\epsilon=0.2$, $\delta=0.1$ and $\varepsilon=0.1$. Each blue triangle marks an estimation of the expectation value with the corresponding number of measurement shots, and for each number of measurement shots, we repeat the estimation for $300$ times. The dashed gray line marks the exact expectation value, and the space between the two dashed orange lines is the zone where an estimation is within the set accuracy.}
    \label{fig:sample_to_estimate}
\end{figure}

\color{black}
\section{Entry Extraction}\label{app:entry_ext}
Entry extraction maps extract entry-wise information of an input quantum state and mix these entries with zeros to form an output matrix for further processing. Below is the mathematical definition of an entry extraction map.
\begin{definition}\label{def:entry extraction}
    Suppose $\cI = \{ i_k \}_{k=0}^{d'-1} \subseteq [d]$ is an indexed set with increasing order. For a set of indexed pairs $A \subseteq \{ (j, k) \mid j \leq k,~ i_j, i_k \in \cI \}$, the corresponding entry extraction map $\cE$ is defined as
\begin{align}
\begin{aligned}
    \cE : \CC^{d \times d} &\to \CC^{d' \times d'} : H = \sum_{j, k = 0}^{d - 1} H_{jk} \ketbra{j}{k}_d  \\
    &\mapsto \sum_{(j, k) \in A \cup A^*} H_{i_j i_k} \ketbra{j}{k}_{d'}
\end{aligned}
,\end{align}
    where $A^* = \{ (k, j) ~|~ (j, k) \in A \}$, $\ketbra{j}{k}_d \in \CC^{d \times d}$ and  $\ketbra{j}{k}_{d'} \in \CC^{d' \times d'}$.
\end{definition}
Note that the Choi representation of $\cE$ is
$J^\cE = \sum_{j, k = 0}^{d - 1} \ketbra{j}{k}_d \otimes \cE(\ketbra{j}{k}_d) = \sum_{(j, k) \in A \cup A^*} \ketbra{i_j}{i_k}_d \otimes \ketbra{j}{k}_{d'}
$.
Then $J^\cE = \lrp{J^\cE}^\dag$ by the symmetry of $A \cup A^*$, implying that $\cE$ is Hermitian-preserving.
\end{document}